\documentclass{article}
\usepackage{amsfonts}
\usepackage{amssymb}
\usepackage{amsmath}
\usepackage{graphicx}
\usepackage{fullpage}

\input{tcilatex}
\begin{document}

\title{Approaching Gaussian Relay Network Capacity in the High SNR Regime:
End-to-End Lattice Codes}
%\author{ \IEEEauthorblockN{Yun Xu}
%\IEEEauthorblockA{Dept. of Electrical Eng.    \\
%    Yale University    \\
%    New Haven, CT 06511    \\
%    Email: yun.xu@yale.edu}
%\and \IEEEauthorblockN{Edmund Yeh}
%\IEEEauthorblockA{Dept. of Electrical \& Computer Eng.    \\
%    Northeastern University    \\
%    Boston, MA 02115    \\
%    Email: eyeh@ece.neu.edu}
%\and \IEEEauthorblockN{Muriel M\'{e}dard}
%\IEEEauthorblockA{Dept. of Electrical Eng. \& Computer Sci.    \\
%    Massachusetts Institute of Technology    \\
%    Cambridge, MA 02139    \\
%    Email: medard@mit.edu} }
\author{Yun Xu\thanks{Department of Electrical Engineering, Yale University, New Haven, CT 06511. Email: yun.xu@yale.edu}
\and Edmund Yeh\thanks{Department of Electrical and Computer Engineering, Northeastern University, Boston, MA 02115.  Email: eyeh@ece.neu.edu}
\and Muriel M\'{e}dard\thanks{Department of Electrical Engineering and Computer Science,  MIT, Cambridge, MA 02139. Email: medard@mit.edu}}
\date{}
\maketitle

\begin{abstract}
We present a natural and low-complexity technique for achieving the capacity of the
Gaussian relay network in the high SNR regime. 
Specifically, we propose the use of end-to-end structured lattice codes with the  
amplify-and-forward strategy, where the source uses a nested lattice code to encode the
messages and the destination decodes the messages by lattice decoding. All
intermediate relays simply amplify and forward the received
signals over the network to the destination. 
We show that the end-to-end lattice-coded amplify-and-forward scheme
approaches the capacity of the layered Gaussian relay network 
in the high SNR regime.   Next, we extend our scheme to non-layered Gaussian 
relay networks under the amplify-and-forward scheme, which 
can  be viewed as a Gaussian intersymbol interference (ISI) channel.  
Compared with other schemes, our approach is significantly simpler and
requires only the end-to-end design of the lattice precoding and decoding.
It does not require any knowledge of the network topology or the individual
channel gains. 
\end{abstract}

%\author{ \IEEEauthorblockN{Yun Xu}
%\IEEEauthorblockA{Dep. of Electrical Eng.    \\
%    Yale University    \\
%    New Haven, CT 06511    \\
%    Email: yun.xu@yale.edu}
%\and \IEEEauthorblockN{Edmund Yeh}
%\IEEEauthorblockA{Dep. of Electrical \& Computer Eng.    \\
%    Northeastern University    \\
%    Boston, MA 02115    \\
%    Email: eyeh@ece.neu.edu}
%\and \IEEEauthorblockN{Muriel M\'edard}
%\IEEEauthorblockA{Dep. of Electrical Eng. \& Computer Sci.    \\
%    Massachusetts Institute of Technology    \\
%    Cambridge, MA 02139    \\
%    Email: medard@mit.edu} }
%\maketitle

\sloppy

\section{Introduction}

Finding the capacity of Gaussian relay networks with one source, one
destination, and a set of relays, has been a long-standing open problem in
network information theory. The relay channel was first
investigated in the seminal work of Cover and EI Gamal \cite{cov79}. More
recently, Kramer {\em et al.}~considered transmission techniques for larger
Gaussian relay networks, {\em e.g.}, the amplify-and-forward,
decode-and-forward, and compressed-and-forward schemes \cite{kra05}.
Avestimeher {\em et al.} presented a deterministic model for Gaussian relay networks
and proposed the quantize-map-and-forward strategy \cite{tse11}. It has been
shown that the quantize-map-and-forward strategy can achieve a rate within a
constant number of bits from the information-theoretic cut-set bound for
Gaussian relay networks, which is independent of the channel gains and the
operating SNR \cite{tse11}. In \cite{dig12}, Ozgur and Diggavi incorporated lattice
codes, lattice quantization, and lattice-to-lattice mapping into the
quantize-map-and-forward scheme. It was shown that the lattice-based
quantize-map-and-forward scheme can still achieve the capacity of Gaussian
relay networks within a constant gap.  While offering strong performance in terms
of achievable rates, the schemes presented in~\cite{tse11,dig12},
involve considerable operational complexity at intermediate relays.

As pointed out in \cite{med12}, in wireless communication settings, signals
simultaneously transmitted from different sources add, leading to the
receiver obtaining a superposition of these signals, scaled by the channel
gains. Since the relays are not interested in the messages sent by the
source, they do not necessarily have to decode, compress or quantize the messages.
Since the relays have already observed the
sum of the signals, in some settings a natural strategy would be to simply amplify and
forward without explicitly dealing with the noise.

A multihop amplify-and-forward scheme with random encoding and decoding
was proposed in \cite{med12}. In this strategy, the message sent by the
source is propagated over many intermediate nodes (relays) and possibly over
multiple hops. All relays exploit the interference and forward the received
signals over the network to the destination. It has been shown
in \cite{med12} that the achievable rate of the multihop amplify-and-forward
scheme approaches the capacity of the Gaussian relay network when the SNR
at the destination is sufficiently high.  
%Importantly, it was
%pointed out in \cite{med12} that high channel gains do not necessarily lead
%to the high SNR regime in a multihop network, unlike in a point-to-point
%channel.

In this paper, we propose the use of end-to-end structured lattice codes with the  
amplify-and-forward strategy, where the source uses a nested lattice code to encode the
messages and the destination decodes the messages by lattice decoding. All
intermediate relays simply amplify and forward the received
signals over the network to the destination.  Relative to the random coding approach of
\cite{med12}, the use of structured lattice codes significantly reduces system complexity
by making possible computationally tractable encoding and decoding.
Furthermore, the use of end-to-end lattice codes implies that we require 
little information concerning the network topology, or individual
channel gains.  Instead, we require only the end-to-end channel response,
which can be obtained by using probing signals.

We show that the end-to-end lattice-coded amplify-and-forward scheme
approaches the capacity of the layered Gaussian relay network under the
high-SNR condition presented in \cite{med12}.   This result is facilitated
by the key observation that a Gaussian layered relay network under amplify-and-forward 
is equivalent to a point-to-point Gaussian channel.  Next, we extend our scheme
to non-layered Gaussian relay networks under the amplify-and-forward scheme, which 
can crucially be viewed as a Gaussian intersymbol interference (ISI) channel.  
%For this case, the
%interleaving/deinterleaving mechanism and the water-filling technique are
%requir
Our lattice-coded amplify-and-forward scheme is simpler than the
lattice-based quantize-map-and-forward scheme proposed in \cite{dig12},
since it does not require lattice quantization and lattice-to-lattice
mapping at relay nodes. For layered networks, only end-to-end design of
the precoding and decoding with nested lattice codes is required.    Thus,
the end-to-end lattice-coded amplify-and-forward scheme is 
a natural and low-complexity technique for achieving the capacity of the
Gaussian relay network in the high SNR regime.

The nested lattice code was
originally proposed by Erez {\em et al.} in \cite{zam02}, \cite{zam04} and \cite%
{zam05-2}. In \cite{zam04}, it was shown that the nested lattice code with
lattice decoding can achieve the capacity of the additive white Gaussian noise (AWGN)
channel at any SNR. Erez {\em et al.} showed that the power-constrained AWGN
channel can be transformed into the modulo-lattice additive noise (MLAN)
channel by minimum mean-square error (MMSE) scaling along with dithering.
The capacity of the MLAN channel, achieved by uniform inputs,
becomes the capacity of the AWGN channel in the limit of large
lattice dimension.

In \cite{zam02} and \cite{zam05-2}, Erez {\em et al.} extended their techniques to
the AWGN dirty-paper channel. They obtained the achievable rate at any SNR
by incorporating MMSE scaling. It was shown that with an appropriate choice
of the lattice, the achievable rate approaches the capacity of the AWGN
dirty-paper channel as the dimension of the lattice goes to infinity. These
results provide an information-theoretic framework used in~\cite{zam02}
to study precoding for the Gaussian ISI channel.
Erez {\em et al.} showed that when combined with the techniques of
interleaving/deinterleaving and water-filling, nested lattice precoding
and decoding can achieve the capacity of the Gaussian ISI channel \cite%
{zam02}.

The remainder of this paper is organized as follows. The network
model is presented in Section II. The fundamental properties of lattices and
nested lattice codes are summarized in Section III. The main result on the
lattice-coded multihop amplify-and-forward strategy is presented in Section
IV. Section V extends the analysis to non-layered networks. Section VI
concludes the paper.

%%%%%%%%%%%%%%%%%%%%%%%%%%%%%%%%%%%%%%%%%%%%%%%%%%%%%%%%%%%%%%%%%%%%%%%%%%%%%%%%%

\section{Network Model}

\subsection{Layered Network}

We first focus on the layered network in which each path from the source to
the destination has the same number of hops. We denote the layer $l$ by $%
\mathcal{L}_{l}$, $l=0,1,\ldots ,L$. Assume that the source $s$ is located
at layer $\mathcal{L}_{0}$, and the destination $d$ at layer $\mathcal{L}%
_{L} $. We denote the number of relays at layer $\mathcal{L}_{l}$ by $n_{l}$%
, and thus $\sum_{l=1}^{L-1}n_{l}=N$. In a layered network, the input-output
relationship is simple due to the fact that all copies of a source message
transmitted on different paths arrive at the destination simultaneously. An
example of a layered network is shown in Figure $\ref{fig-LN}$.

%%%%%%%%%%%%%%%%%%%%%%%%%%%%%%%%%%%%%%%%%%%%%%%%%%%%%%
\begin{figure}[h]
\centering
\includegraphics[height=2in]{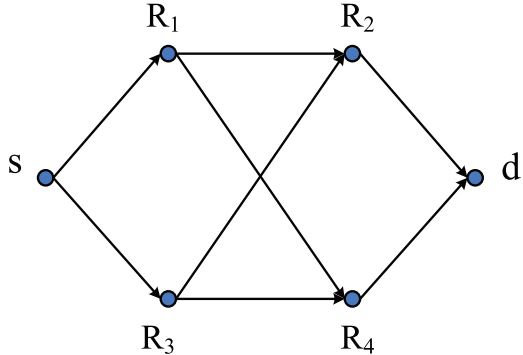}
\caption{Example of a Layered Network}
\label{fig-LN}
\end{figure}
%%%%%%%%%%%%%%%%%%%%%%%%%%%%%%%%%%%%%%%%%%%%%%%%%%%%%%

Now consider a layered Gaussian relay network consisting of a single source $%
s$, a single destination $d$, and a set of $N$ relays. The communication
link from node $i$ to node $j$ has a nonnegative real channel gain,
represented by $h_{ij}\in \mathbb{R}_{+}$. The channel output at node $j\neq
s$ is
\begin{equation}
y_{j}=\sum_{i\in \mathcal{N}\left( j\right) }h_{ij}x_{i}+z_{j}.  \label{WGN}
\end{equation}%
where $x_{i}$ is the channel input at node $i$, $z_{j}$ is the real Gaussian
noise with zero mean and unit variance, and $\mathcal{N}\left( j\right) $
denotes the set of nodes that can transmit to node $j$ with a direct link,
i.e., $\mathcal{N}\left( j\right) =\left\{ i:h_{ij}>0\right\} $. Note that
the links are assumed to be directed so that $i\in \mathcal{N}\left(
j\right) $ does not imply $j\in \mathcal{N}\left( i\right) $. We assume that
there is an average power constraint at each node:%
\begin{equation}
E\left[ X_{i}^{2}\right] \leq P_{i}  \label{PC}.
\end{equation}%
The source $s$ wishes to send a message from a message set $\mathcal{W}%
=\left\{ 1,\ldots ,2^{nR}\right\} $ to the destination $d$ with
transmission rate $R$. The encoding function at the source is given by $%
X_{s}^{n}=f\left( W\right) ,W\in \mathcal{W}$, and a decoding function at
the destination $d$ is given by $\hat{W}=g\left( Y_{d}^{n}\right) $. A $%
\left( R,n\right) $ code consists of a message set $\mathcal{W}$, an
encoding function at the source, and a decoding function at the destination.
The average error probability of the $\left( R,n\right) $ code is given by $%
P_{e}=\Pr [\hat{W}\neq W]$. A rate $R$ is said to be achievable if for any $%
\epsilon >0$, there exists a $\left( R,n\right) $ code such that $P_{e}\leq
\epsilon $ for a sufficiently large $n$.

%%%%%%%%%%%%%%%%%%%%%%%%%%%%%%%%%%%%%%%%%%%%%%%%%%%%%%%%%%%%%%%%%%%%%%%%%%%%%%%%%

\subsection{High SNR Regime}

As in \cite{med12}, we are interested in the scenario in which all relays
forward the data with large enough power to guarantee that the total
propagated noise at the destination by multihop amplify-and-forward is low
enough. Assume that each node $i$ transmits with the average power $P_{i}$
given by $\left( \ref{PC}\right) $. The power received at relay $j\in
\mathcal{L}_{l},l=1,\ldots ,L-1$, is then determined by%
\begin{equation*}
P_{R,j}=\left( \sum_{i\in \mathcal{L}_{l-1}}h_{ij}\sqrt{P_{i}}\right) ^{2},\
\ \ j\in \mathcal{L}_{l}
\end{equation*}
and the power received at the destination $d$ is given by
\begin{equation}
P_{d}=\left( \sum_{i\in \mathcal{L}_{L-1}}h_{id}\sqrt{P_{i}}\right) ^{2}.
\label{Pd}
\end{equation}

%\begin{definition} \cite{med12}
%A wireless network is in the high SNR regime if%
%\begin{equation}
%\min_{j\in \mathcal{L}_{l}}P_{R,j}\geq \frac{1}{\delta },\text{ \ \ }%
%l=1,\ldots ,L-1  \label{HSNR}
%\end{equation}%
%for some small $\delta >0$.
%\end{definition}

As in \cite{med12}, we consider a high SNR regime where for some
small $\delta >0$, the transmit powers of the relays satisfy
\begin{equation}
\min_{j\in \mathcal{L}_{l}}P_{R,j}\geq \frac{1}{\delta },\text{ \ \ }%
l=1,\ldots ,L-1.  \label{HSNR}
\end{equation}

We then assume that $P_{d}$ remains a constant as $\delta \rightarrow 0$, so
that the Multiple-Access Channel (MAC)
at the destination is a bottleneck for the data transmission.\footnote{
Note that for network capacity, the worst case occurs when the bottleneck
is at the MAC at the destination.  In this case, the noise is propagated over
more hops than in any other case.}  The MAC cut-set
bound is given by
\begin{equation}
C_{MAC}=\frac{1}{2}\log \left( 1+P_{d}\right).  \label{CMAC}
\end{equation}

%%%%%%%%%%%%%%%%%%%%%%%%%%%%%%%%%%%%%%%%%%%%%%%%%%%%%%%%%%%%%%%%%%%%%%%%%%%%%%%%%

\section{Lattices and Nested Lattice Codes}

In this section, we briefly review some basic properties of lattices and
nested lattice codes. A more extensive discussion can be found in references 
such as \cite{zam02} and \cite{zam04}. A lattice $\Lambda $ is a discrete subgroup
of the Euclidean space $\mathbb{R}^{n}$. If $\lambda _{1}$ and $\lambda _{2}$
are two elements of a lattice $\Lambda $, then the sum $\lambda
_{1}+\lambda _{2}$ and the additive inverse $-\lambda _{1}$ are also
elements of $\Lambda $. Any lattice can be written in terms of its
generating matrix $G$:%
\begin{equation*}
\Lambda =\left\{ \lambda =G\boldsymbol{x}:\boldsymbol{x}\in \mathbb{Z}%
^{n}\right\}.
\end{equation*}%

We can then define the nearest neighbor quantizer associated with $\Lambda $
by%
\begin{equation*}
Q_{\Lambda }\left( \boldsymbol{x}\right) =\func{argmin}_{\lambda \in \Lambda
}\left\Vert \boldsymbol{x}-\lambda \right\Vert.
\end{equation*}%
The Voronoi region of a lattice point $\lambda \in \Lambda $ is the set of
all points that quantize to it. The fundamental Voronoi region $\mathcal{V}$
is the set of all points that quantize to the origin, i.e.,
\begin{equation*}
\mathcal{V}=\{\boldsymbol{x}:Q_{\Lambda }\left( \boldsymbol{x}\right) =%
\boldsymbol{0}\}.
\end{equation*}%
Define the modulo-$\Lambda $ operation corresponding to $\mathcal{V}$ as
\begin{equation*}
\boldsymbol{x}\text{ mod }\Lambda =\boldsymbol{x}-Q_{\Lambda }\left(
\boldsymbol{x}\right).
\end{equation*}

The second moment of a lattice $\Lambda $ is defined by%
\begin{equation*}
\sigma _{\Lambda }^{2}=\frac{1}{nVol\left( \mathcal{V}\right) }\int_{%
\mathcal{V}}\left\Vert \boldsymbol{x}\right\Vert ^{2}d\boldsymbol{x},
\end{equation*}%
and the normalized second moment of a lattice $\Lambda $ is defined by
\begin{equation*}
G\left( \Lambda \right) =\frac{1}{n\left[ Vol\left( \mathcal{V}\right) %
\right] ^{1+2/n}}\int_{\mathcal{V}}\left\Vert \boldsymbol{x}\right\Vert ^{2}d%
\boldsymbol{x},
\end{equation*}%
where $Vol\left( \mathcal{V}\right) $ is the volume of $\mathcal{V}$.

Two lattices $\Lambda _{1}$ and $\Lambda _{2}$ are said to be nested if $%
\Lambda _{1}\subseteq \Lambda _{2}$, where $\Lambda _{1}$ is called the
coarse lattice and $\Lambda _{2}$ the fine lattice. Denote by $\mathcal{V}%
_{1}$ and $\mathcal{V}_{2}$ the fundamental Voronoi regions of $\Lambda _{1}$
and $\Lambda _{2}$, respectively. The coding rate is defined by
\begin{equation*}
R=\frac{1}{n}\log \left[ \frac{Vol\left( \mathcal{V}_{1}\right) }{Vol\left(
\mathcal{V}_{2}\right) }\right].
\end{equation*}%
The points in the set
\begin{equation*}
\mathcal{C}=\Lambda _{2}\cap \mathcal{V}_{1}
\end{equation*}%
are called the coset leaders of $\Lambda _{1}$ relative to $\Lambda _{2}$.
For each $\boldsymbol{c}\in \mathcal{C}$, the shifted coarse lattice $%
\Lambda _{1,c}=\boldsymbol{c}+\Lambda _{1}$ is called a coset of $\Lambda
_{1}$ relative to $\Lambda _{2}$.

The use of high-dimensional nested lattice code is justified by the
existence of asymptotically good lattices. We consider two types of goodness
as introduced in \cite{zam02} and \cite{zam04}.

(1) \textit{Good for AWGN Channel Coding:} For any $\epsilon >0$ and
sufficiently large $n$, there exists an $n$-dimensional lattice $\Lambda $
with the volume of the fundamental Voronoi region $Vol\left( \mathcal{V}%
\right) <2^{n\left[ h\left( Z\right) +\epsilon \right] }$, where $Z$ is
Gaussian noise with variance $\sigma _{Z}^{2}$, and $h\left( Z\right) =%
\frac{1}{2}\log \left( 2\pi e\sigma _{Z}^{2}\right) $ is the differential
entropy of $Z$, such that%
\begin{equation*}
P_{e}=\Pr \left[ Z\notin \mathcal{V}\right] <\epsilon.
\end{equation*}

\textit{(2) Good for Source Coding: }For any $\epsilon >0$ and sufficiently
large $n$, there exists an $n$-dimensional lattice $\Lambda $ whose
normalized second moment $G\left( \Lambda \right) $ satisfies%
\begin{equation*}
\log \left( 2\pi eG\left( \Lambda \right) \right) <\epsilon.
\end{equation*}

%%%%%%%%%%%%%%%%%%%%%%%%%%%%%%%%%%%%%%%%%%%%%%%%%%%%%%%%%%%%%%%%%%%%%%%%%%%%%%%%%

\section{Lattice-Coded Multihop Amplify-and-Forward}

In~\cite{med12}, it is shown that the simple multihop amplify-and-forward scheme with
random coding can approximately achieve the unicast capacity of a Gaussian
relay network in the high SNR regime. In this paper, we propose the use of
structured nested lattice codes in conjunction with the multihop amplify-and-forward scheme. Choose a
pair of nested lattices $\left( \Lambda _{1},\Lambda _{2}\right) $, $\Lambda
_{1}\subseteq \Lambda _{2}$, with the coding rate
\begin{equation*}
R=\frac{1}{n}\log \left[ \frac{Vol\left( \mathcal{V}_{1}\right) }{Vol\left(
\mathcal{V}_{2}\right) }\right] \geq R_{LAF},
\end{equation*}%
where $R_{LAF}$ is the rate achieved by the lattice-coded
amplify-and-forward scheme defined by $\left( \ref{AFR}\right) $. We choose
the coarse lattice $\Lambda _{1}$ to be good for source coding, with the
second moment $\sigma _{\Lambda _{1}}^{2}=P_{s}$, where $P_{s}$ is the
average power constraint of the source node $s$. We choose the fine lattice $%
\Lambda _{2}$ to be good for AWGN channel coding. Let $Q_{\Lambda _{2}}$
denote the nearest neighbor quantizer of the fine lattice $\Lambda _{2}$. We
apply the scheme proposed in \cite{zam02}, \cite{zam04} and \cite{zam05-2}
to the layered Gaussian relay network as follows:

(1) \textit{Source:} the source $s$ maps the message $W\in \mathcal{W}$
uniformly at random to a coset leader of $\Lambda _{1}$ relative to $\Lambda _{2}$:%
\begin{equation*}
\boldsymbol{t}\in \mathcal{C}=\Lambda _{2}\cap \mathcal{V}_{1}.
\end{equation*}%
Then a random dither vector $\boldsymbol{u}$ is generated uniformly over $%
\mathcal{V}_{1}$, i.e., $\boldsymbol{u}\sim Unif\left( \mathcal{V}%
_{1}\right) $. Given the message $W$, the source encoder sends%
\begin{equation}
\boldsymbol{x}_{s}=[\boldsymbol{t}+\boldsymbol{u}]\text{ mod }\Lambda _{1}.
\label{SEC}
\end{equation}

(2) \textit{Relay: }each relay $i\in \mathcal{L}_{l},l=1,\ldots ,L-1$,
performs the multihop amplify-and-forward scheme%
\begin{equation}
\boldsymbol{x}_{l,i}=\beta _{i}\boldsymbol{y}_{l,i},\text{ \ \ \ }i\in
\mathcal{L}_{l}  \label{AF}
\end{equation}%
where the amplification gain is chosen as%
\begin{equation}
\beta _{i}=\frac{\sqrt{P_{i}}}{\sqrt{\left( 1+\delta \right) P_{R,i}}},\text{
\ \ }i\in \mathcal{L}_{l}.  \label{AmGain}
\end{equation}%
In~\cite{med12}, it is shown that the power constraint $\left( \ref{PC}\right) $ at
each node $i$ is satisfied by choosing the amplification gain in $\left( \ref%
{AmGain}\right) $.  Also shown in \cite{med12} are the following two
results on the propagated noise.

\begin{lemma}
\cite{med12} At any node $i\in \mathcal{L}_{l}$, the noise propagated from
all nodes in layer $\mathcal{L}_{l-k},k=1,\ldots ,l-1$, via the multihop
amplify-and-forward scheme in the high SNR regime has the power%
\[
P_{z,i}^{l-k}\leq \frac{\delta P_{R,i}}{\left( 1+\delta \right) ^{k}}.
\]
\end{lemma}

\begin{corollary} \label{TPN}
\cite{med12} The total noise propagated to the destination $d \in \mathcal{L}_{L}$ has the power
\begin{equation}
P_{z,d}=\sum_{k=1}^{L-1}P_{z,d}^{L-k}=\delta P_{d}\sum_{k=1}^{L-1}\frac{1}{%
\left( 1+\delta \right) ^{k}}\leq L\delta P_{d}.  \label{TPNP}
\end{equation}
\end{corollary}

(3) \textit{Destination:} the destination computes
\begin{equation*}
\boldsymbol{\hat{y}}_{d}=Q_{\Lambda _{2}}\left( \alpha \boldsymbol{y}_{d}+%
\boldsymbol{u}\right) \text{ mod }\Lambda _{1},
\end{equation*}%
where
\begin{equation*}
\alpha =\frac{\gamma }{1+\gamma },\text{ \ \ }\gamma =\frac{P_{d}}{\left(
1+\delta \right) ^{L-1}P_{z,d}}.
\end{equation*}%
%and $P_{z,d}$ is given by $\left( \ref{TPNP}\right) $.

The following theorem is the main result for the lattice-coded
amplify-and-forward scheme.

\begin{theorem}
In a layered relay network $\left( \ref{WGN}\right) $ in the high SNR regime
defined by $\left( \ref{HSNR}\right) $, the lattice-coded multihop
amplify-and-forward achieves the rate%
\begin{equation}
R_{LAF}\geq \frac{1}{2}\log \left[ 1+\frac{1}{\left( 1+\delta \right) ^{L-1}}%
\frac{P_{d}}{1+L\delta P_{d}}\right].   \label{AFR}
\end{equation}
\end{theorem}

\begin{proof}
As in \cite{med12}, if the amplification gain at each relay is chosen as $%
\left( \ref{AmGain}\right) $, the received signal at the destination can be
written as%
\begin{equation}
y_{d}=\hat{h}_{d}x_{s}+\hat{z}_{d}+z_{d},  \label{GC}
\end{equation}%
where $\hat{z}_{d}$ the total propagated noise, $z_{d}$ is the noise at
the destination, and
\begin{equation}
\hat{h}_{d}=\frac{\sqrt{P_{d}}}{\sqrt{P_{s}\left( 1+\delta \right) ^{L-1}}}.
\label{hd}
\end{equation}%
By $\left( \ref{GC}\right) $ and $\left( \ref{hd}\right) $, the received
signal power at the destination is%
\begin{equation*}
\hat{P}_{d}=\frac{P_{d}}{\left( 1+\delta \right) ^{L-1}}.
\end{equation*}%
By Corollary $\ref{TPN}$, the power of the total propagated noise $\hat{z}%
_{d}$ is%
\begin{equation*}
P_{z,d}=\delta P_{d}\sum_{k=1}^{L-1}\frac{1}{\left( 1+\delta \right) ^{k}}%
\leq L\delta P_{d}.
\end{equation*}%
Therefore the SNR at the destination satisfies%
\begin{equation}
SNR\geq \frac{1}{\left( 1+\delta \right) ^{L-1}}\frac{P_{d}}{1+L\delta P_{d}}.
\label{GSNR}
\end{equation}%
In other words, the received signal $y_{d}$ can be viewed as the output of
the AWGN channel characterized by $\left( \ref{GC}\right) $ with the SNR
given by $\left( \ref{GSNR}\right) $. The capacity, or equivalently the
achievable rate via amplify-and-forward, is then given by
\begin{eqnarray*}
R_{LAF} &=&\frac{1}{2}\log \left( 1+SNR\right) \\
&\geq &\frac{1}{2}\log \left[ 1+\frac{1}{\left( 1+\delta \right) ^{L-1}}%
\frac{P_{d}}{1+L\delta P_{d}}\right].
\end{eqnarray*}%
It is shown in \cite{zam02} and \cite{zam04} that if we choose the coarse lattice to
be good for source coding and the fine lattice $\Lambda _{2}$ to be good for
AWGN channel coding, nested lattice codes can achieve the capacity of the
AWGN Gaussian channel when the dimension, or equivalently, the length of the
codewords $n$, tends to infinity. Hence, the above lattice-coded
amplify-and-forward scheme can achieve the rate $R_{LAF}$.
\end{proof}

Note that as $\delta \rightarrow 0$, the rate achieved by the lattice-coded
multihop amplify-and-forward scheme in $\left( \ref{AFR}\right) $
approaches the MAC cut-set bound $\left( \ref{CMAC}\right) $, and the
unicast capacity of the Gaussian relay network.

%%%%%%%%%%%%%%%%%%%%%%%%%%%%%%%%%%%%%%%%%%%%%%%%%%%%%%%%%%%%%%%%%%%%%%%%%%%%%%%%%

\section{Extensions}

\subsection{Non-layered Networks}

In layered networks, each path from the source to the
destination has the same number of hops, so that all copies of the source
message transmitted on different paths arrive at the destination with the
same delay.  In non-layered networks, however, copies of the source message
may arrive at the destination with different delays through different paths.
Assume that the number of hops (length) in the longest path is $L\geq 1$. We
can then classify all paths from the source to the destination according to
the path length%
\begin{equation*}
P_{l}=\left\{ \text{paths of length }l\right\}
\end{equation*}%
Assume that the number of paths of length $l$ is $K_{l},l=1,\ldots ,L$. As
shown in \cite{med12}, the received signal at the destination $d$ at time $t$
is then given by%
\begin{eqnarray*}
y_{d}\left( t\right) &=&h_{0}x_{s}\left( t\right) +\sum_{j\in
P_{1}}h_{j,1}x_{s}\left( t-1\right) +\ldots \\
&&+\sum_{j\in P_{L}}h_{j,L}x_{s}\left( t-L\right) +z_{e}\left( t\right),
\label{eq:ISI}
\end{eqnarray*}%
where $h_{0}$ is the channel gain on the direct link from the source to the
destination, and $h_{j,l}$ is the equivalent channel gain of path $j$
in the set $P_{l}$. Note that $h_{j,l}$ depends on the network topology, and
contains the accumulated channel gains and amplification gains on the
source-destination path $j$.  Finally, $z_{e}\left( t\right) $ denotes the total
propagated noise at the destination.

From~\eqref{eq:ISI}, we see that under the amplify-and-forward scheme,
the non-layered Gaussian relay network is equivalent to a Gaussian
ISI channel:
\begin{equation}
y_{d}\left( t\right) =\sum_{l=0}^{L}h_{l}x_{s}\left( t-l\right) +z_{e}\left(
t\right),  \label{GISI}
\end{equation}%
where $x_{s}\left( t-l\right) ,l=0,\ldots ,L$, are the inputs to the
Gaussian ISI channel, $h_{l}=$ $\sum_{j\in P_{l}}h_{j,l}$ represents the ISI
coefficient, and $y_{d}\left( t\right) $ stands for the received samples.
The additive Gaussian noise is denoted by $z_{e}\left( t\right) $.

We now focus on the feedforward MMSE decision feedback (MMSE-DFE) equalizing
filter for the Gaussian ISI channel \cite{for95}, \cite{for95-2}. The output
of the MMSE-DFE feedforward filter can be written as
\begin{equation}
r\left( t\right) =x_{s}\left( t\right) +s\left( t\right) +n\left( t\right),
\label{DP}
\end{equation}%
where $s\left( t\right) =\sum_{l=1}^{L}\hat{h}_{l}x_{s}\left( t-l\right)$
%\begin{equation*}
%s\left( t\right) =\sum_{l=1}^{L}\hat{h}_{l}x_{s}\left( t-l\right)
%\end{equation*}%
is the post-cursor intersymbol interference, and $\hat{h}_{l}$ and $n\left(
t\right) $ represent the ISI coefficients and the sampled noise at the
output of the MMSE-DFE feedforward filter, respectively.

The SNR associated with the MMSE-DFE filter is
defined by \cite{for95}%
\begin{equation*}
SNR_{MMSE-DFE}=\frac{E\left[ X_{s}^{2}\left( t\right) \right] }{E\left[
N^{2}\left( t\right) \right] },
\end{equation*}%
and the capacity of the Gaussian ISI channel is given by \cite{for95-2}%
\begin{equation}
C_{ISI}=\frac{1}{2}\log \left( 1+SNR_{MMSE-DFE}\right).  \label{CISI}
\end{equation}

If the encoder knows the entire post-cursor intersymbol interference vector
before transmission, as mentioned in \cite{zam02}, the input-output
relationship given by $\left( \ref{DP}\right) $ can be viewed as the
Gaussian dirty-paper channel whose capacity is given by $\left( \ref{CISI}%
\right) $. Based on that observation, \cite{zam02} proposed a coding
strategy for the Gaussian ISI channel, in which the MMSE-DFE feedback
equalizing filter is replaced by nested lattice precoding, as described in
the dirty paper case. In the interleaver, the messages are encoded row by
row and are transmitted column by column. When a message which comprises the
$j$th row of the interleaver is to be encoded, the post-cursor interfering
symbols belong to the codewords for messages which have been already
encoded, similar to the dirty paper scenario \cite{zam02}.  In~\cite{zam02},
it is shown that nested lattice precoding with interleaving/deinterleaving
and waterfilling can achieve the capacity of the Gaussian ISI channel given by $%
\left( \ref{CISI}\right) $. 
Waterfilling is required because the sampled noise $n(t)$ in \eqref{DP} 
may not be white Gaussian. 
If we incorporate nested lattice precoding into the multihop 
amplify-and-forward scheme, the achievable rate can thus be obtained as the 
capacity of the corresponding Gaussian ISI channel. 
The lattice-coded multihop amplify-and-forward scheme is shown in 
Figure $\ref{fig-ISI}$.

Unlike the case for layered networks, in order to implement nested lattice precoding
with interleaving for non-layered Gaussian wireless relay networks, it is
necessary to know the channel gains as manifested in the ``ISI coefficients."
In the absence of such knowledge, techniques such as blind equalization may
be needed to preserve the performance of our scheme.

If the encoder does not know the entire post-cursor ISI vector before transmission, 
we have to use the original MMSE-DFE technique proposed in \cite{for95} with nested 
lattice encoding and decoding. In that case, however, the decoded messages must be 
fed back. Thus, any decoding error will affect the performance of the MMSE-DFE 
feedback filter, and hence the subsequent decoding process. This may lead to the
decoding error propagating over multiple symbols.

%%%%%%%%%%%%%%%%%%%%%%%%%%%%%%%%%%%%%%%%%%%%%%%%%%%%%%
\begin{figure}[h]
\centering
\includegraphics[height=3.5in]{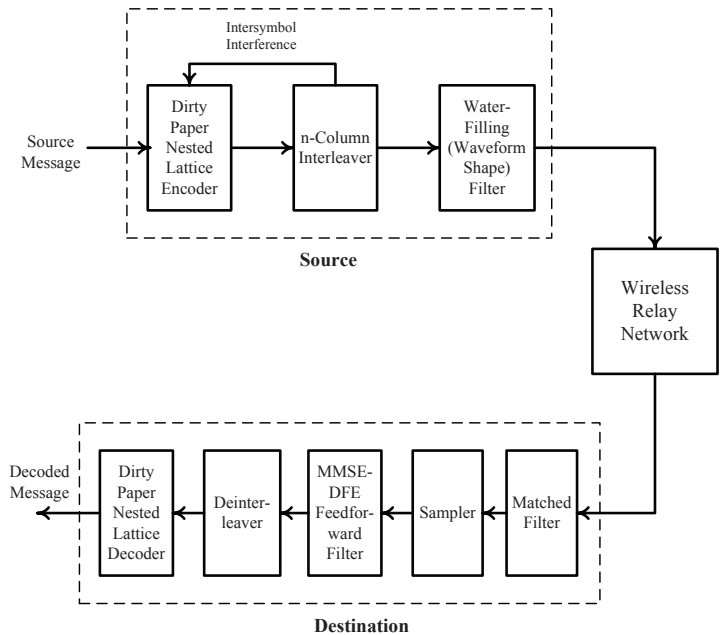}
\caption{Lattice-Coded Multihop Amplify-and-Forward Scheme}
\label{fig-ISI}
\end{figure}
%%%%%%%%%%%%%%%%%%%%%%%%%%%%%%%%%%%%%%%%%%%%%%%%%%%%%%

%The corresponding transmission scheme is as follows:
%
%(1) \textit{Source:} The source node $s$ maps the message $W\in \mathcal{W}$
%to a coset leaders of $\Lambda _{1}$ relative to $\Lambda _{2}$:
%\begin{equation*}
%\boldsymbol{t}\in \mathcal{C}=\Lambda _{2}\cap \mathcal{V}_{1}
%\end{equation*}%
%Then a random dither vector $\boldsymbol{u}$ is generated uniformly over $%
%\mathcal{V}_{1}$, i.e., $\boldsymbol{u}\sim Unif\left( \mathcal{V}%
%_{1}\right) $. Given the message $W$, the source encoder sends%
%\begin{equation*}
%\boldsymbol{x}_{s}=[\boldsymbol{t}-\theta \boldsymbol{s}+\boldsymbol{u}]%
%\text{ mod }\Lambda _{1}
%\end{equation*}%
%where
%\begin{equation*}
%\theta =\frac{P_{s}}{P_{s}+P_{N}}
%\end{equation*}%
%and $P_{N}$ is the power of the corresponding sampled noise at the output of
%the MMSE-DFE feedforward filter.
%
%(2) \textit{Relay: }Each relay $i$ performs the multihop amplify-and-forward
%scheme%
%\begin{equation*}
%\boldsymbol{x}_{i}=\beta _{i}\boldsymbol{y}_{i}
%\end{equation*}%
%where the amplification gain is chosen by $\left( \ref{AmGain}\right) $.
%
%(3) \textit{Destination:} The decoder at the destination computes
%\begin{equation*}
%\boldsymbol{\hat{y}}_{d}=Q_{\Lambda _{2}}\left( \theta \boldsymbol{y}_{d}+%
%\boldsymbol{u}\right) \text{ mod }\Lambda _{1}
%\end{equation*}

%%%%%%%%%%%%%%%%%%%%%%%%%%%%%%%%%%%%%%%%%%%%%%%%%%%%%%%%%%%%%%%%%%%%%%%%%%%%%%%%%

\section{Conclusion}

In this paper, we considered an end-to-end lattice-coded multihop amplify-and-forward
strategy for Gaussian wireless relay networks in the high SNR regime. When the
power received at all relays are large enough, our strategy performs well
for both layered and non-layered Gaussian relay networks. In the worst case,
the bottleneck of the multihop transmission is at the multi-access channel (MAC)
at the destination. We showed that our strategy approaches the MAC cut-set
bound as the received powers at the relays increase.
The lattice-coded multihop amplify-and-forward scheme is simpler than
the decode-and-forward scheme and the quantize-map-and-forward scheme. Our
scheme requires only end-to-end design: lattice precoding at the
source and decoding at the destination. It does not require any knowledge of
the network topology or the individual channel gains.

%%%%%%%%%%%%%%%%%%%%%%%%%%%%%%%%%%%%%%%%%%%%%%%%%%%%%%%%%%%%%%%%%%%%%%%%%%%%
\bibliographystyle{IEEEtran}
\bibliography{IEEE-lattice}

\end{document}